\documentclass[11pt,oneside]{article}

\usepackage{amssymb,bm}
\usepackage{amsmath,amsthm}
\usepackage{amsfonts}
\usepackage[pagewise,mathlines]{lineno}
\usepackage{verbatim}
\usepackage{psfrag,graphicx}\graphicspath{{./figures/}{./figurespng/}}
\usepackage{amsmath,natbib}
\usepackage{setspace}
\usepackage{subfigure}
\usepackage{geometry,setspace}
\usepackage[page,title]{appendix}
\usepackage{color}
\usepackage{amsfonts}
\usepackage{amsmath}
\usepackage{amssymb}
\usepackage{graphicx}%
\usepackage{lineno}
\usepackage[flushleft]{threeparttable}
\usepackage{subfigure}

\newtheorem{theorem}{Theorem}

\newtheorem{alg}{Algorithm}

\newcommand{\phat}{\widehat{p}}

\newcommand{\ghat}{\widehat g}

\newcommand{\veps}{\varepsilon}

\newcommand{\Jhat}{\widehat J}
\newcommand{\Ehat}{\widehat E}
\renewcommand{\hbar}{\overline h}

\newcommand{\Exp}{{\mathbb{E}}}
\newcommand{\one}{\mathbf{1}}

\newcommand{\vectornorm}[1]{\lVert#1\rVert}
\newcommand{\tu}{\widetilde u}
\newcommand{\SNR}{\text{signal to noise ratio}}
\newcommand{\swarm}{\mathbb{A}}

\begin{document}

\title{Bayesian Inference for Nonlinear Structural Time Series  Models}
\author{Jamie Hall\\{\small School of Economics}\\{\small University of New South Wales}\\{\small jamie1212@gmail.com}
\and Michael K. Pitt\\{\small Economics Department}\\{\small University of Warwick}\\{\small m.pitt@warwick.ac.uk}
\and Robert Kohn\\{\small School of Economics}\\{\small University of New South Wales}\\{\small r.kohn@unsw.edu.au} }
\date{\today}
\maketitle


\begin{abstract}
This article discusses a partially adapted particle filter for estimating
the likelihood of nonlinear structural econometric state space models whose state transition density
cannot be expressed in closed form.  The filter generates the disturbances in
the state transition equation and allows for multiple modes in the conditional disturbance distribution.
The particle filter produces an unbiased estimate of the likelihood
and so can be used to carry out Bayesian inference in a particle Markov chain Monte Carlo
framework. We show empirically that when the signal to noise ratio is high,
the new filter can be much more efficient
than the standard particle filter, in the sense that it requires far fewer particles to give the
same accuracy. The new filter is applied to several simulated and real examples
and in particular to a dynamic stochastic general equilibrium model.



\end{abstract}

\section{Introduction}
For a general state space model the standard particle filter \citep{gordon_novel_1993} gives an unbiased estimate of the likelihood. \cite{andrieu_particle_2010} show that it is possible to use this unbiased estimate within a Markov chain Monte Carlo (MCMC) sampling scheme to carry out Bayesian inference for the parameters of the state space model. They call such a sampling scheme particle MCMC (PMCMC). PMCMC is particularly useful for Bayesian inference when the state space model is nonlinear or non-Gaussian so that the Kalman filter cannot be used. However, when the signal to noise ratio of the model is high, i.e., when the observation vector gives a very informative measurement on some combination(s) of elements of the state vector, the standard particle filter becomes a computationally inefficient importance sampler (see \cite{pitt_filtering_1999}).

For many models, this problem can be solved by using adapted  particle filters as in~\cite{pitt_filtering_1999}, which are more efficient as importance samplers. In addition, \cite{pitt_properties_2012}  show empirically that fully adapted particle filters used within PMCMC may require far fewer particles than the standard particle filter to achieve the same accuracy.

While using adapted particle filters can be much more efficient than the standard particle filter, most
adapted particle filters require that we can evaluate the state transition density. In important cases, this density is not easily available in closed form. This applies, in particular, to dynamic stochastic general equilibrium (DSGE) models, which are currently widely used in applied macroeconomics. This is in contrast to the standard particle filter which only requires that we can evaluate the observation density and simulate from the state transition density.

Our article proposes to solve this problem by using a partially adapted particle filter that generates the states by first generating
the disturbances in the state transition equation. The idea that a stochastic process may become more tractable when considered in terms of its innovations has a long history (see \cite{heunis_innovations_2011}). We employ it here because it provides a simple and useful solution to the problem described above. This approach has also been considered in recent independent work by \cite{murray_collapsed_2012}, who use it to estimate biological models with intractable transition densities. Our approach differs in its emphasis on solutions tailored for structural econometrics. Specifically, we demonstrate that a proposal based on a numerical optimisation algorithm and allowing for multiple modes in the disturbances by using mixtures, delivers large efficiency gains when applied to rational-expectations models with high signal-to-noise ratios when compared to the standard particle filter and the filter in \cite{murray_collapsed_2012}.
In \cite{murray_collapsed_2012}, the authors use a sigma-point approximation to the conditional densities of the disturbances.

Two other possible improvements to the particle filter for rational expectations models have been proposed in recent literature. \cite{amisano_euro_2010} demonstrate that the particle filter can be made more efficient if the proposal distribution for the state vector $x_t$ is conditioned on the first two moments of $x_{t-1}$, estimated from the particle swarm in the previous period. Second, \cite{andreasen_non-linear_2011} demonstrates an improvement when the proposal density for each particle is based on a central-difference Kalman filter with a rescaled covariance. While these contributions are valuable, we believe that the algorithm discussed here represents something of an improvement on these methods, since it is able to deliver results using a relatively small number of particles.

Section \ref{section: description of adpf} describes the partially adapted particle filter used in our article
and gives its properties. Section \ref{section: qar example}  demonstrates its performance on simulated data from a simple nonlinear time series model. Section~\ref{section: parameter estimation} demonstrates its application to nonlinear DSGE models: specifically, a neoclassical growth model and a consumption-based asset pricing model. Section \ref{section: conclusion} concludes.

\section{Auxiliary Disturbance Particle Filter}
\label{section: description of adpf}
This section discusses  a particle filter which is effective for certain classes of
models explored in this paper. This filter, which we call the auxiliary
disturbance particle filter (ADPF), is similar in principle to the auxiliary particle
filter but works by sampling the disturbances in the state equation
rather than the states themselves. More specifically, while the auxiliary
particle filter works by building an approximation to the joint density of
$(x_{t-1}^{k},x_{t})^{\prime}$, given the previous observations $y_{1:t-1}
=\{y_{1},...,y_{t-1}\}$ and $y_{t}$, the ADPF  builds  an
approximation to $(x_{t-1}^{k},u_{t})^{\prime}$, given $y_{1:t}$,
where $u_{t}$ is the disturbance term in the state equation. To simplify the
notation in this section, we often omit to show dependence on unknown parameters.

\subsection{General principles}
Consider the state space model with measurement density $p(y_t|x_t)$
and state transition equation
\begin{align}
x_{t} &  =h(x_{t-1}, u_{t}), \label{eq: state transition}
\end{align}
where $u_t$ is an independent sequence, $h(x_{t-1}, u_{t})$ is a nonlinear
function of $x_{t-1}$ and $u_t$.
Given a sample $y_{1:T} = \{y_1, \dots, y_T\}$, and the unknown parameters,
the likelihood is $p(y_{1:T})  =p(y_1)\prod_{t=2}^T  p(y_t|y_{1:t-1}) . $
When the function $h(\cdot)$ is nonlinear,
it is usually impossible to evaluate the likelihood exactly, but we can estimate it
using one of a number of particle filters. The standard particle filter \citep{gordon_novel_1993}
can be used whenever it is possible to evaluate the measurement density $p(y_t|x_t)$
and generate from the state transition equation \eqref{eq: state transition}.
However, the standard particle filter can be quite inefficient in the sense that it produces a likelihood
estimate with a large variance compared to more sophisticated particle filters \citep[see][]{pitt_properties_2012}.
\cite{pitt_filtering_1999} suggest a class of auxiliary particle filters that can be much more efficient than the standard particle filter, but most of these require the evaluation of the
density of the state transition equation. However,
for  many cases that are of interest to us, it is infeasible to evaluate the state transition equation density which means that existing auxiliary particle filters cannot be used.

The ADPF attempts to overcome this problem by approximating the state space model with measurement density $p(y_t|x_t)$ and state transition equation~\eqref{eq: state transition} as follows. Suppose that the density $g(y_{t+1}|x_t)$ approximates $p(y_{t+1}|x_t)$ and
the density $g(u_{t+1}|y_{t+1},x_t)$ approximates the density $p(u_{t+1}|y_{t+1},x_t)$ and that
 we can evaluate $p(y_{t+1}|x_t)$ and $g(u_{t+1}|y_{t+1},x_t)$ and generate from
$g(u_{t+1}|y_{t+1},x_t)$.
The choice of this approximate density may be model-specific, although we discuss a general implementation in Section~\ref{section: adpf implementation}.

To explain intuitively how the ADPF is constructed, suppose that
 $\{(x_t^k,\pi_t^k), k=1, \dots, N\}$ is a swarm of particles generated
from an approximation to $p(x_t|y_{1:t})$, so that $x_t^k$ has associated weight $\pi_t^k$.
Now, define
\begin{align*}
\phat(dx_t|y_{1:t}) & = \sum_{k=1}^N \pi_{t}^{k} \delta_{x_t^k}(dx_t )\,
\end{align*}
where $\delta_s(dx)$ is the Dirac delta distribution centered at $s$. We shall show how to construct $\phat(dx_{t+1}|y_{1:t+1}) $.

Suppose we wish to
estimate $\Exp (m(x_{t+1})|y_{1:t+1}) $, where $m(\cdot) $
is a function of $x_{t+1}$, assuming that the expectation exists.
Define the functional
\begin{align}
J_{t+1}(m) & = \int m(x_{t+1}) p(y_{t+1}|x_t,u_{t+1})p(x_t, u_{t+1}|y_{1:t}) dx_t du_{t+1} , \notag \\
& = \int m(x_{t+1}) p(y_{t+1}|x_t,u_{t+1})p(u_{t+1})p(x_t|y_{1:t}) dx_t du_{t+1},\notag \\
& = \int m(x_{t+1})p(y_{t+1}|x_t) p(u_{t+1}|y_{t+1}, x_t) p(x_t|y_{1:t})  dx_t du_{t+1},
\label{eq: exact m}
\end{align}
using the identity $p(y_{t+1}|x_t,u_{t+1})p(u_{t+1}) = p(y_{t+1}|x_t) p(u_{t+1}|y_{t+1}, x_t)$.
Then, it is straightforward to check that $\Exp (m(x_{t+1})|y_{1:t+1}) = J_{t+1}(m)/J_{t+1}(\one)$, where
$\one$ is the unit function, and $p(y_{t+1}|y_{1:t}) = J_{t+1}(\one)$.

We approximate $J_{t+1}(m)$ by replacing $p(x_t|y_{1:t})dt $ in \eqref{eq: exact m} by
$\phat(dx_t|y_{1:t})$ to obtain,
\begin{align*}
J_{t+1}(m)
 & \approx  \int m(x_{t+1}) p(y_{t+1}|h(x_{t}; u_{t+1})) p(u_{t+1})\phat (x_{t}|y_{1:t}) du_{t+1} dx_{t}\\
 & = \int m(x_{t+1}) \frac{ p(y_{t+1}|h(x_{t}; u_{t+1})) p(u_{t+1})}
 {  g(u_{t+1}|y_{t+1},x_t) g( y_{t+1}|x_t) } g( y_{t+1}|x_t) g(u_{t+1}|y_{t+1},x_t)
 \phat (dx_{t}|y_{1:t}) du_{t+1} \\
 & = \left (  \sum_{i=1}^{N}\omega_{t|t+1}^{i}   \right )
 \int m(x_{t+1})\frac{ p(y_{t+1}|h(x_{t}, u_{t+1}) p(u_{t+1})}
 {  g(u_{t+1}|y_{t+1},x_t) g( y_{t+1}|x_t) }  g(u_{t+1}|y_{t+1},x_t) \ghat_N (dx_{t}|y_{1:t+1}) du_{t+1}
 \intertext{where}
\ghat_N (dx_{t}|y_{1:t+1})  & = \sum_{k=1}^N \pi_{t|t+1}^{k} \delta_{x_t^k}(dx_t )
 , \pi_{t|t+1}^{k}=\frac
{\omega_{t|t+1}^{k}}{\sum_{i=1}^{N}\omega_{t|t+1}^{i}}
\quad  \text{with} \quad \omega_{t|t+1}^{k}
 =g(y_{t+1}|x_{t}^{k})\pi_{t}^{k} .
\end{align*}
Suppose that $u_{t+1}^k \sim g(u_{t+1}|y_{t+1},x_t^k) $ and $x_{t+1}^k = h(x_t^k,  u_{t+1}^k)$, and put
\begin{align*}
\omega_{t+1}^{k} & = \frac{ p(y_{t+1}|x_{t+1}^k ) p(u_{t+1}^k)}
 {  g(u_{t+1}^k|y_{t+1},x_t^k) g ( y_{t+1}|x_t^k) }, \quad \pi_{t+1}^{k}=\frac{\omega_{t+1}^{k}} %
{\sum_{i=1}^{N}\omega_{t+1}^{i}}\ .
\end{align*}
We obtain the estimate
\begin{align}\label{eq: expected value m }
\Jhat _{t+1}(m) & = \left (  \sum_{i=1}^{N}\omega_{t|t+1}^{i}   \right )
\left ( \sum_{i=1}^{N}\omega_{t+1}^{i} \right) \sum_{k=1}^N m(x_{t+1}^k) \pi_{t+1}^{k}\\
\text{with} \quad \Jhat_{t+1}(1) & = \left (  \sum_{i=1}^{N}\omega_{t|t+1}^{i}   \right )
\left ( \sum_{i=1}^{N}\omega_{t+1}^{i} \right),
\end{align}
so that
\begin{align*}
\Ehat(m(x_{t+1})|y_{1:t+1}) & = \sum_{k=1}^N m(x_{t+1}^k) \pi_{t+1}^{k} \quad \text{and} \quad
\phat(y_{t+1}|y_{1:t}) & = \left (  \sum_{i=1}^{N}\omega_{t|t+1}^{i}   \right )
\left ( \sum_{i=1}^{N}\omega_{t+1}^{i} \right).
\end{align*}

This suggests that we take $\{(x_{t+1}^k, \pi_{t+1}^k), k=1, \dots, N\}$ as the swarm
of particles that we use to approximate $p(x_{t+1}|y_{1:t+1})$ and define
\begin{align*}
\phat(dx_{t+1}|y_{1:t+1}) & = \sum_{k=1}^N \pi_{t+1}^{k} \delta_{x_{t+1}^k}(dx_t )\ .
\end{align*}

The following algorithm formally describes the  ADPF and is
initialized with a sample $x_{0}^{k}\sim p(x_{0})$ with mass $\pi_0^k = 1/N$ for
$k=1,...,N$.

\begin{alg}
\label{alg:apdf}
\textsl{For }$\mathit{t=0,..,T-1}$,
\textsl{given samples }$x_{t}^{k}\sim p(x_{t}|y_{1:t})$\textsl{\ with mass }$\pi_{t}^{k}$ \textsl{for }$k=1,...,N.$

\begin{enumerate}
\item \textsl{For }$k=1:N,$\textsl{\ compute }$\omega_{t|t+1}^{k}%
=g(y_{t+1}|x_{t}^{k})\pi_{t}^{k},$ \ \ \ \ $\pi_{t|t+1}^{k}=\frac
{\omega_{t|t+1}^{k}}{\sum_{i=1}^{N}\omega_{t|t+1}^{i}}.$

\item \textsl{For }$k=1:N,$\textsl{\ sample }$\widetilde{x}_{t}^{k}\sim
\sum_{i=1}^{N}\pi_{t|t+1}^{i}\delta_{x_t^i}(dx_{t}).$

\item \textsl{For }$k=1:N,$\textsl{\ sample }$u_{t+1}^{k}\sim g(u_{t+1}%
|\widetilde{x}_{t}^{k};y_{t+1})$ and put $x_{t+1}^k = h(x_t^k; u_{t+1}^k)$.

\item {\normalsize \textsl{For }$k=1:N,$\textsl{\ compute}%
\[
\omega_{t+1}^{k}=\frac{p(y_{t+1}|x_{t+1}^{k})p(u_{t+1}^{k})}{
g(y_{t+1}|\widetilde{x}_{t}^{k})g(u_{t+1}^{k}|\widetilde{x}_{t}%
^{k};y_{t+1})},\text{ \ \ \ \ \ \ }\pi_{t+1}^{k}=\frac{\omega_{t+1}^{k}}%
{\sum_{i=1}^{N}\omega_{t+1}^{i}}.
\]
}
\end{enumerate}
\end{alg}

The estimate of the likelihood corresponding to the ADPF is
\begin{align} \label{eq: likel estimate adpf}
\phat (y_{1:T}) & = \prod_{t=1}^T
\left (  \sum_{i=1}^{N}\omega_{t-1|t}^{i}   \right )
\left ( \sum_{i=1}^{N}\omega_{t}^{i} \right)
\end{align}

By \cite{andrieu_particle_2010}) we can write this estimate of the likelihood as $\phat_N (y_{1:t}|\theta, \zeta ) $ where $\zeta$ consists of a set canonical random variables that are used to construct the estimate and that have density $p_N(\zeta)$.
Without loss of generality we can assume that elements of
$\zeta$ are uniform independent random variates. The estimated likelihood also depends on the number of particles $N$.

\begin{theorem} \label{Thm: unbiased}
The estimate $\phat_N (y_{1:T}| \zeta )$ is unbiased in the sense that
\begin{align*}
\int \phat_N(y_{1:t}| \zeta ) p_N(\zeta) d\zeta & = p(y_{1:T}) .
\end{align*}
The proof is in Appendix~\ref{Proof: Thm 1}.
\end{theorem}

\section{Particle filter performance on a first order nonlinear autoregressive model}
\label{section: qar example}

\subsection{The model}
Consider the following univariate nonlinear time series model,
\begin{align}
y_t &= x_t + \sigma_{\varepsilon} \varepsilon_t \\
x_t &= \phi x_{t-1} + \sigma_u \left( u_t + \delta u_t^2 \right)
\label{qar1_state_eqn}
\end{align}
where $\varepsilon_t$ and $u_t$ are \textit{iid} standard normal random variables. We choose this model because it is one of the simplest nonlinear extensions to a well-understood linear model.
When $\sigma_\veps$ is small it is also one of the simplest examples of the class of structural economic models that we consider below.
The parameter $\delta$ controls the degree of nonlinearity in the model; with $\delta=0$ the model is a first order autoregressive (AR(1)) model  with observation noise. When $\delta$ exceeds about 0.5, the behavior of the model
becomes noticeably different to that of a linear model. For a comprehensive analysis of this general class of models, see \cite{aruoba_new_2011}.

\subsection{ADPF implementation and testing}
\label{section: adpf implementation}

We choose a normal distribution with mean $\phi x_{t-1} + \sigma_u \delta$ and variance $\sigma_\varepsilon^2 + \sigma_u^2(1 + 2 \delta^2)$ for the approximating density $g(y_t | x_{t-1})$ because it matches the first two moments of $y_t$, conditional on $x_{t-1}$.
To estimate $g(u_t | y_t, x_{t-1}^k)$, we proceed as follows. Let $\ell^k(u) = \log p(y_t | u, x_{t-1}^k) \phi(u; 0,1)$ for given $x_{t-1}^k$ and $y_t$, where $\phi(u; a, b^2)$ is the univariate
 normal density in $u$ with mean $a$ and variance $b^2$.
 We numerically maximise $\ell^k(u)$ over $u$, subject to equation~(\ref{qar1_state_eqn}), and initialise with a draw from $\phi(u; 0, 2^2)$,
which is reasonable because $u$ has density $\phi(u; 0,1)$. Let $\tu_t^k$ be the mode of $ \ell^k(u)$.
 Then we obtain a normal approximation $\phi(u;\tu_t^k , \Delta_t^k) $ to $ \ell^k(u)$, where
 $\Delta^k_t =  -\left( \partial^2 \ell^k(u) / \partial u^2  \right)^{-1}$ evaluated at $u = \tu_t^k$.
Appendix~\ref{app: optimization} gives more details on the optimization procedure used here, and in the next example.

The normal approximation must be renormalised to adjust for the presence of multiple local modes. Multiple modes occur because the law of motion is a nonlinear polynomial, so that $\ell^k(u)$ is also nonlinear; in this case, a quartic equation. In other words, a given observation $y_t$ might have been generated by more than one possible value of $u_t$. For that reason, a simple normal approximation is inadequate. To see this, consider the limit as $\sigma_\varepsilon \to 0$. In this case, the model becomes deterministic, and almost any realised level of $y_t$ is then consistent with two values of $u_t$. Suppose, in a particular case, we label those two values $u^{(1)}_t$ and $u^{(2)}_t$. The algorithm, as described so far, would place a mass of 1 on the value of $u_t$ that the numerical minimiser found; that is, if it happened to find the first mode, it would sample $u_t$ from $ g(u_t | y_t, x_t^{k}) = \phi(u_t; u^{(1)}_t , \Delta_t^{(1)})$,  where $\Delta_t^{(1)} = -\left( \partial^2 \ell^k(u) / \partial u^2  \right)^{-1}$ evaluated at $u = u_t^{(1)}$.

Figure~\ref{fig: bimodal densities} illustrates this point. The left-hand panel plots $x_t$ as a function of $u_t$, that is, equation~(\ref{qar1_state_eqn}). This plot is assumes a previous value $x_{t-1}=0$, and parameter values $\delta=0.5$, $\sigma_\varepsilon=0.2$, $\sigma_u = 1$. Any value of the underlying state $x_t > -0.5$ is consistent with two possible values of $u_t$. For example, a value of $x_t=0.5$ could have been generated by $u_t\approx 0.41$ or $u_t\approx-2.4$. The right-hand panel plots an unnormalised version of $\ell^k(u)$, conditional on an observation $y_t=0.5$. The distribution is obviously bimodal; with the noise variance $\sigma_\varepsilon$ small in this case, the two modes of $\ell^k(u)$ are close to 0.41 and $-2.4$.

With $\sigma_\varepsilon \to 0$, the simple maximisation approach described above would place a weight of 1 on, say, the value of $u=0.41$, ignoring the other possibility. This is not an accurate approximation to the true $p(u_t | y_t, x_t^{k}))$, and will produce what appears to be a biased estimate of the log-likelihood. Although the ADPF is unbiased no matter what proposal density $g(u_t | y_t, x_{t-1}^{k})$ is used, in cases with a high signal to noise ratio, a very large number of particles are required to counteract the problem described in the text, which defeats the purpose of using the filter. Instead, a better approximation is given by the mixture density
\begin{equation}
\label{eq: mixture-density-1}
g(u_t | y_t, x_t^{k}) = \phi(u^{(1)}_t; 0,1) \phi (u_t; u^{(1)}_t , \Delta_t^{(1)}) + \phi(u^{(2)}_t; 0,1) \phi(u_t; u^{(2)}_t , \Delta_t^{(2)}).
\end{equation}
In this simple model, it is feasible to search for both modes $u_t^{(1)}$ and $u_t^{(2)}$,
but we develop a more general approach, with the goal of constraining the
required computation time as the dimension of the model increases.
Given a set of estimates $\tu_t^k$, obtained as described above,
we form a proposal density for each disturbance $u^j_t, j = 1, \dots, N$,
by taking an equally-weighted mixture of those $\tu^k$ that generate
a value for $y_t$ within 3 standard deviations of the observed value. That is,
\begin{equation}
\label{eq: mixture density for u}
g(u^j_t | y_t, x_t^{j}) \propto \sum_{i=1}^N {\chi}\left[ \left( \frac{y_t - \phi x^j_{t-1} - \sigma_u \left( u^i_t + \delta (u^i_t)^2 \right)}{\sigma_\varepsilon} \right)^2 \leq 3^2 \right] \phi \left(u_t^j; \tu^i_t , \Delta^{(i)} \right)
\end{equation}
where $\chi\left[\omega\right]$ is the characteristic or indicator function of event $\omega$. This takes advantage of the fact that values of $x^k_{t-1}$ tend to be clustered, so that a value of $\tu_t^k$ found for a particular $x^k_{t-1}$ has a good chance of working well for another $x^j_{t-1}$, in the sense of implying a mean value of $y_t$ close to the observed one. This method of generating a proposal generalises to more complex models, such as the DSGE example discussed in Section \ref{section: parameter estimation} below. Note that the proposal density is reweighted by the true density in step 4 of the ADPF algorithm. Thus, there is no need to ensure that (for instance) the proposal density puts the correct mass on different possible values of $u$, as in equation (\ref{eq: mixture-density-1}). All that is required is that each possible value of $u$ has a reasonable chance of being used in the proposal.

\begin{center}
{Figure \ref{fig: bimodal densities} about here}
\end{center}

We performed a simulation study to compare the performance of the ADPF with the standard SIR particle filter and
the `CUPF1' algorithm described in \cite{murray_collapsed_2012}. That algorithm has a similar structure to the ADPF, except that the first-stage proposal density $g(y_t | x_{t-1}^k)$ and the second-stage density $g(u_t^k | y_t, x_{t-1}^k)$ are both given by the Unscented Kalman Filter \citep{wan_unscented_2000}, run individually for each particle. We briefly describe the  Unscented Kalman Filter, with a full description given by
\cite{wan_unscented_2000} and \cite{van_der_merwe_unscented_2001}. In the Unscented Kalman Filter
the prior mean $x_{t-1}$ is propagated through the model's law of motion, along with a set of `sigma points' $x_{t-1} \pm \sqrt{\lambda P_{t-1}}_{(i)}$, where $P_{t-1}$ is the prior state covariance augmented with the covariance of the disturbances and noise terms, $\lambda$ is a parameter of the algorithm, $\sqrt{\cdot}$ denotes a matrix square root, and $P_{(i)}$ is the $i^th$ row of the matrix.\footnote{More precisely, $\lambda$ is a function of parameters $\alpha$, $\beta$ and $\kappa$. We used the values suggested in \cite{wan_unscented_2000}, namely $\alpha=10^{-3}$, $\beta=2$, $\kappa=0$. We experimented with different values, but found that this did not appear to alter the results substantially.} After the sigma points are propagated through the law of motion and the observation equation, we obtain an accurate estimate of the mean and covariance of the state and disturbances, conditional on the observation $y_t$.

We evaluated the filters on four different parameter settings: with either $\delta = 0.7$ (a high nonlinearity case) or $\delta = 0.1$ (low nonlinearity), and with either $\sigma_e = 0.01$ (a high signal to noise (SNR) ratio) or $\sigma_e = 1.0$ (low SNR). In all cases, we set $\sigma_u = 1$ and $\phi=0.6$.

For each test, we simulated a single dataset of 50 observations. We chose this number of observations because it roughly corresponds to the length of a quarterly macroeconomic data series. Using 1000 replications, we calculated the median log-likelihood estimated by the filters on each dataset, along with the interquartile range of their estimates, as well as the standard deviation of their log-likelihood estimates. Additionally, we used a standard particle filter with 1,000,000 particles to estimate the true value of the loglikelihood, which we used to estimate the bias of the  loglikelihood estimates of the other filters.

\subsection{Results}

The performance of the ADPF  is comparable to that of the standard SIR filter in the two cases with low single-to-noise ratio. In both scenarios, the standard deviation of the log-likelihood estimates from the ADPF with 50 particles is between those from the standard particle filter with 100 and 500 particles. This is similar to the performance of the fully adapted particle filter, which makes no substantial improvement on the standard particle filter when the \SNR{} is relatively low.

The situation is different with a high SNR, reported in Tables~\ref{low-non-high-snr-table} and \ref{high-non-high-snr-table}. In these cases, the ADPF  proposal draws from $p(u_t | y_t, x_{t-1})$ are much more useful than the draws from the proposal
$p(x_t | x_{t-1})$ used by the standard particle filter, because the observation $y_t$ is highly informative about the current position of $x_t$ (and therefore of $u_t$). As a result, the precision of the ADPF with 50 particles is similar to the standard particle filter's with more than 7500 when the model is markedly nonlinear, and more than 15000 in the approximately linear case. Additionally, the CUPF1 variation does not appear to be well adapted to this class of model, perhaps because more than the first two moments are required for a good approximation to the target density.

\begin{center}{Tables~\ref{low-non-high-snr-table} and  \ref{high-non-high-snr-table} about here.} \end{center}

Tables~\ref{low-non-low-snr-bias-table} to \ref{high-non-high-snr-bias-table} report the estimated bias and variance in  the log-likelihood estimates
for the four combinations of nonlinearity and \SNR.
The asymptotic analysis in \cite{pitt_properties_2012} suggests that the log-likelihood estimates should have a bias approximately equal to $-0.5$
times their variance, and that the variance should decrease in proportion to the number of particles used. Our simulations of the simple quadratic AR(1) model are broadly consistent with these expectations, with the predictions  borne out well in the low \SNR{} cases. The high \SNR{} cases reported in Tables~\ref{low-non-high-snr-bias-table} and \ref{high-non-high-snr-bias-table} appear to be consistent with the predictions of \cite{pitt_properties_2012} as the number of particles becomes large, though the uncertainty around these estimates is larger in these cases.

\begin{center}{Tables~\ref{low-non-low-snr-bias-table} to \ref{high-non-high-snr-bias-table}  about here.}\end{center}

\section{Parameter estimation}
\label{section: parameter estimation}
\subsection{Example 1: Neoclassical growth model}
\subsubsection{Model}
 This section considers a basic neoclassical growth model. We choose this model because it is a useful and simple benchmark for solving and estimating DSGEs, used for example in \cite{schmitt-grohe_solving_2004} and \cite{gomme_second_order_2011}.
 The model is based on the decisions of a representative household,
 which chooses between consumption $c_t$ and investment in next period's capital stock $k_t$.
 The household's goal is to maximise discounted lifetime utility, given by
 \[ U = \sum_{t=0}^\infty \beta^t \log c_t \] subject to a feasibility constraint,
 \begin{equation}
\label{dsge-feasibility-eqn}
c_t + k_t = A_t k_{t-1}^\alpha + (1-\delta)k_{t-1}\ ,
\end{equation}
where $\delta \in [0,1]$, and a productivity shock
\begin{equation}
\label{dsge-shock-eqn}
\log A_t = \rho \log A_{t-1} + \epsilon_t  \hspace{1cm} \epsilon_t \sim N(0, \sigma^2_\epsilon)\ .
\end{equation}
where $\rho \in (0,1)$. The solution of the model consists of equations \eqref{dsge-feasibility-eqn} and \eqref{dsge-shock-eqn} plus a consumption Euler equation,
\begin{equation}
\label{dsge-euler-eqn}
c_t^{-1} = \beta \mathbb{E}_t \left\{ c_{t+1}^{-1}
\left[ \alpha A_{t+1} k_t^{\alpha-1} + 1 - \delta  \right]
 \right\} .
\end{equation}
Here, $\mathbb{E}_tX$ denotes the model-consistent expectation of $X$, conditional on information available at time $t$.

This solution can be converted to a Markov process on the assumption of rational expectations as in \cite{klein_using_2000}. If the depreciation rate $\delta$ is below one, the conversion cannot be expressed in closed form, and some type of approximation must be used. We chose a second-order approximation, using the methods described in \cite{klein_using_2000} and \cite{gomme_second_order_2011}.
The output of these methods is a law of motion for the vector $x_t = \left( c_t, k_t, a_t \right)^\prime$ of the form
\begin{align}
\nonumber
x_t &= d + E x_{t-1} + F \epsilon_t + \left(\mathrm{I}_3 \otimes x_{t-1}^\prime \right) G x_{t-1}  \\
& \hspace{1cm} + \left(\mathrm{I}_3 \otimes x_{t-1}^\prime \right) H \epsilon_{t}  + \left(\mathrm{I}_3 \otimes \epsilon_{t}^\prime \right) J \epsilon_{t} .
\label{dsge-lom}
\end{align}
The reduced-form coefficient matrices $d$, $E$, $F$, $G$, $H$ and $J$ are functions of the structural parameters, but
must be calculated numerically as we do not have analytical expressions for them.

As is standard in the DSGE literature, we add the \lq measurement error\rq{} $\nu_t$ to the observation equation, in order to avoid stochastic singularity and for computational convenience, i.e.
\begin{equation}
\label{dsge-obs-eqn}
y_t = \left[1,0,0\right] x_t + \nu_t \hspace{1cm} \nu_t \sim N(0, \sigma^2_\nu)\ .
\end{equation}
For linear DSGE models,  $\nu_t$ is usually assumed to be small, with $\sigma^2_\nu$ many orders of magnitude smaller than $\sigma^2_\epsilon$. This assumption is sometimes relaxed for second-order estimation in order to reduce the sampling error of the standard particle filter. In some cases, such as the asset pricing model considered below, measurement error is an important part of the model, and the variance of the measurement noise is comparable to that of the innovations in the model's law of motion. Here, we maintain the assumption of a high \SNR, setting $\sigma^2_\nu$ to $10^{-8}$.
We set the rest of the parameters to fairly standard calibrated values \citep[see e.g.][]{gomme_second_order_2011, schmitt-grohe_solving_2004}. Specifically, we set $\beta = 0.99, \alpha =  1/3, \rho = 0.8,
\delta =  0.05$ and $\sigma_\epsilon^2 = 0.02^2$.

\subsubsection{Estimation}
To evaluate the performance of the ADPF in estimation, we simulated a data series of 50 observations using equations \eqref{dsge-lom} and \eqref{dsge-obs-eqn}. Again, we chose this length observations because it is of the same order of magnitude as a macroeconomic time-series. We then use the adaptive random walk Metropolis-Hastings \citep{roberts_optimal_2001} to take 100,000 draws from the parameter vector. We fix $\beta$ at 0.99. This is standard, as it is difficult to identify $\beta$. The vector of unknown parameters is $\theta=\left( \alpha, \rho, \delta, \sigma_\epsilon \right)$. Table~\ref{growth-priors-table} summarizes the priors
on the structural parameters, which are set relatively loosely to assist identification.

\begin{center}
{Table~\ref{growth-priors-table} about here. }\end{center}

We initialised the Metropolis-Hastings chain at the maximum likelihood estimate obtained from a first-order approximation of the model via the Kalman filter. We chose this initialisation method because we observed that the standard deviation of the log of the estimate of the likelihood obtained by the ADPF increased significantly in some areas of the support of $\theta$ away from the true values,  making it difficult for the Metropolis-Hastings algorithm to converge. Additionally, initialising the MCMC chain in this way has been the practice for second-order DSGE estimation using the standard particle filter, as in \cite{fernandez-villaverde_how_2007}
and \cite{amisano_euro_2010}. The Metropolis-Hastings proposal covariance matrix was initialised to a diagonal matrix of small positive values, with adaptation beginning after 100 draws.

We repeated this procedure using different numbers of particles in the particle filter and the ADPF (using the same simulated data). For each estimation run, we report the Metropolis-Hastings acceptance rate, the inefficiency, and the computation time. For each component of the parameter vector, the inefficiency is calculated as $IF = 1 + 2 \sum_{j=1}^{L^\star} \widehat{\rho}_j$, where $\widehat{\rho}_j$ is the estimated autocorrelation of the parameter iterates at lag $j$. If $K$ is the sample size used to compute $\widehat{\rho}_j$, then the maximum lag length is set to $L^\star = \min\{1000,L\}$, with $L$ being the lowest index $j$ such that $\left| \widehat{\rho}_j \right| < 2/\sqrt{K}$. 

Since the actual wall-clock estimation time depends heavily on the details of a particular implementation, we report instead a measure of computation time calculated as the number of evaluations of the model's law of motion required to produce one effectively independent draw from a given parameter. Thus, if $N$ particles use the law of motion an average of $k$ times each in an MCMC run with an inefficiency (as described above) given by $IF$, then the computation time is taken to be $CT = k \times N \times IF$. \cite{pitt_properties_2012} measure the  computation time as  $N \times IF$. Here, we ensure a fair comparison with the standard particle filter by penalising the ADPF for the multiple function evaluations required for optimising $\ell^k(u_t)$. We estimate $k$ by keeping a tally of the number of times the law of motion subroutine was called during estimation. For this model, the value of $k$ was around 16 per particle per observation.

In implementing the ADPF, we use the same approach as  in section~\ref{section: adpf implementation}. For the first-stage proposal density $g(y_{t+1}|x_{t})$, we use a normal distribution matching the first two moments of $y_{t+1}$ conditional on $x_{t}$, which can be calculated from equations (\ref{dsge-lom}) and (\ref{dsge-obs-eqn}).

Specifically, substituting (\ref{dsge-lom}) into (\ref{dsge-obs-eqn}), dropping the negligible measurement error term $\nu_t$, and taking expectations,  the mean of $y_{t}$ conditional on $x_{t-1}$ is
\begin{equation}
\mu_{y,t} = d_1 + E_1 x_{t-1}
	+ x_{t-1}^\prime G_{1} x_{t-1} +  \sigma^2_\epsilon J_{1}\ ,
\label{eq: conditional mean y}
\end{equation}
where $d_1$ and $E_1$ are the first element and row of $d$ and $E$, $G_1$ is the upper $3 \times 3$ blocks of $G$, and $J_1$ is the first element of $J$. Similarly, the conditional variance of $y_t$ is
\begin{equation}
\sigma_{y,t} = \sigma^2_\epsilon F_1^2 + 2 \sigma^2_\epsilon F_1 x_{t-1}^\prime H_1 +
\sigma^2_\epsilon \left( x_{t-1}^\prime H_1 \right)^2 + \sigma^4_\epsilon J_1^2 \ ,
\label{eq: conditional variance y}
\end{equation}
where $F_1$ is the first element of $F$, $H_1$ is the first row of $H$, and $J_1$ is the first row of $J$. Note that if we use a first-order approximation to the solution of the DSGE model, then $G = H = J = 0$, and the mean (\ref{eq: conditional mean y}) and variance (\ref{eq: conditional variance y}) are equal to the one-step prediction mean and variance from the Kalman filter (conditioned on a given value for $x_{t-1}$).

\subsubsection{Results and Analysis}

As expected, the estimated parameter values from all filters are very similar. However, there are notable differences in efficiency and computing time. Table \ref{acc-ineff-results-table} reports the Metropolis-Hastings acceptance rates for different numbers of particles used in the standard particle filter and the ADPF. It also shows the inefficiencies for each component of the parameter vector. As the number of particles used increases, so that the estimates of the loglikelihood become more precise, the acceptance rate increases. This is true for both the standard particle filter and the ADPF. Broadly speaking, the ADPF performs about as well with 50 particles as does the standard particle filter with several thousand particles. Conversely, to approach the performance of the ADPF with 300 particles, the standard particle filter must use about 10,000.

The differences in inefficiency are also reflected in the
estimates of computing time, which are reported in Table~\ref{ct-results-table}. As explained above, the estimates of computing time are a function both of the number of computations required to generate a given number of draws of the parameter vector, and also of the inefficiency of those draws. While the inefficiencies decrease steadily as the number of particles was increased, the computing time requires a tradeoff between higher numbers of particles (which reduces inefficiency) and lower numbers of particles (which directly reduces computing time). The optimal computing times occur with around 1,500 particles for the standard particle filter and 30 for the ADPF. The  computing time of the ADPF
is roughly one fifth of the standard particle filter. Since the implementation of the ADPF leaves scope for optimisation or parallel computing, this relative performance could be improved further in practice.\footnote{The optimisation step of the APDF algorithm can be performed in parallel for larger problems. Additionally, we penalise the ADPF for every evaluation of the law of motion; but, conditional on $x_{t-1}^k$, only half of equation~\eqref{dsge-lom} needs to be recalculated for a given value of $\epsilon_t$.}

Table \ref{ct-results-table} also shows the variance of each filter's loglikelihood estimates, evaluated by taking 75 repeated loglikelihood estimates at the true value of $\theta$. These results are broadly consistent with the asymptotic analysis in \cite{pitt_properties_2012}, which suggests that the optimal computing time would be attained when the loglikelihood variance is around 0.81. In the case of the ADPF, the optimal computing times occur when the variance is slightly higher than that. We conjecture that this is because the analysis in \cite{pitt_properties_2012} assumes that a perfect Metropolis-Hastings proposal distribution is available.

\begin{center}{Tables~\ref{acc-ineff-results-table} and \ref{ct-results-table} about here.} \end{center}

\subsection{Example 2: Asset pricing with habits}

\subsubsection{The Model}

This section demonstrates the full-information estimation of a structural asset pricing model, specifically, a consumption-based model with external habits
\citep{campbell_by_1999}. Previously, this type of model has been taken to the data by matching moments, e.g. \cite{campbell_by_1999}), using GMM , e.g. \cite{hyde_dont_2005}, or a linear approximation, e.g. \cite{bouakez_habit_2005}). Here, we estimate the likelihood directly.

The model assumes that the representative agent's consumption process is
\begin{equation}
\Delta \log C_t = g + \nu_t\ ,
\label{campbell-consumption}
\end{equation}
where $\nu \sim N(0, \sigma^2)$. The agent's utility function is given by
\begin{equation}
U_t = E_t \sum_{t=0}^\infty \beta^t \frac{\left(C_t - X_t\right)^{1-\gamma}}{1-\gamma}\ ,
\label{campbell-utility}
\end{equation}
where $X_t$ is the (external) habit stock, interpreted as the minimum level of consumption required to maintain a well-defined utility (i.e., the household must ensure that $C_t > X_t$). The surplus consumption ratio $S_t$ and the deviation
$\widetilde s_t$ of $\log S_t$ from its mean $\overline S$ are defined by
\[ S_t = \frac{C_t - X_t}{C_t} \quad \text{and} \quad  \widetilde{s}_t = \log S_t - \log \overline{S}\]
The law of motion of $\widetilde{s}_t$ is assumed to be
\begin{equation}
\widetilde{s}_t = \phi \widetilde{s}_{t-1} + \left(\frac{1}{\overline{S}}\sqrt{1 - 2\widetilde{s}_{t-1}} -1\right)\nu_t\ ,
\label{campbell-habit}
\end{equation}
where the disturbance $\nu_t$ is the same as the consumption innovation in equation (\ref{campbell-consumption}), and the steady-state level of $S_t$ is
\begin{equation}
\overline{S} = \sigma \sqrt{\frac{\gamma}{1-\phi}}\ .
\label{campbell-sbar-level}
\end{equation}
See \cite{campbell_by_1999} for discussion of the derivation of equations \eqref{campbell-habit} and \eqref{campbell-sbar-level}.  The ratio $S_t$ is stationary, since the level of the habit stock $X_t$ is constructed to grow at the same rate as $C_t$ in the long run. The nonlinear functional form for $\widetilde{s}_t$ means that habit is a slowly-moving average of recent consumption in `normal times' (when the habit stock is close to its mean) but responds more sensitively to consumption innovations during `bad times' (when $\widetilde{s}_t$ is low). This nonlinearity allows the model to address both the equity premium puzzle and the risk-free rate puzzle (see \cite{campbell_by_1999} for further discussion).

On that basis, it can be shown that the equilibrium price-dividend ratio of a financial asset   satisfies
\begin{equation}
\frac{P_t}{D_t} = \beta_t \mathbb{E}_t \left[ \exp\left[ \gamma(\widetilde{s}_t - \widetilde{s}_{t+1}) + (1-\gamma)(g + \nu_{t+1}) \right]\left(1 + \frac{P_{t+1}}{D_{t+1}}\right) \right]\ ,
\label{campbell-euler-eqn}
\end{equation}
where $\beta_t$ is the intertemporal discount factor in period $t$. This is the Fundamental Theorem of Finance---the current asset price equals the risk-neutral expectation of next period's asset price and return---using the functional forms implied by equations \eqref{campbell-utility} to \eqref{campbell-sbar-level} \citep{campbell_by_1999}. Since not all changes in the typical household's intertemporal trade-offs can be explained by consumption habits, we choose to perturb this parameter with a shock process that is a random walk in logs,
\begin{equation}
\beta_t = \overline{\beta} \mathrm{e}^{b_t}, \hspace{5mm} b_t = b_{t-1} + \epsilon_t\ ,
\label{campbell-beta-eqn}
\end{equation}
where $\epsilon_t \sim N(0, \sigma^2_\epsilon)$.

Equations (\ref{campbell-consumption}), (\ref{campbell-habit}), (\ref{campbell-euler-eqn}) and (\ref{campbell-beta-eqn}) characterise the model. The observed variables are $\Delta \log C_t$ and $\Delta \log \frac{P_t}{D_t}$, and the observation equations are (\ref{campbell-consumption}) and (\ref{campbell-euler-eqn}). We approximate equation (\ref{campbell-euler-eqn}) using a third-order Taylor series expansion in $\widetilde{s}$ at $\widetilde{s}=0$ because it is the simplest approximation that describes the nonlinear behaviour of the price-dividend ratio adequately. Thus equation (\ref{campbell-euler-eqn}) is approximated as
\begin{equation}
\log \frac{P_t}{D_t} = \Gamma_0 + \Gamma_1 \widetilde{s}_t + \Gamma_2 \widetilde{s}_t^2 + \Gamma_3 \widetilde{s}_t^3\ ,
\label{cambell-third-order-eqn}
\end{equation}
where the $\Gamma_i$ coefficients are functions of the structural parameters (see Appendix~\ref{A: asset pricing approximation}).

\subsubsection{Estimation}

We apply the model to observations of growth in the S{\&}P500 price-dividend ratio and US consumption using
 quarterly observations from 1950 to 2011, a total of 248 datapoints, that are plotted in Figure~\ref{fig: asset data}. The S{\&}P500 series is from \cite{shiller_irrational_2006}, while the consumption series is the seasonally adjusted real personal consumption expenditure series from the Bureau of Economic Analysis (series code PCECC96).

\begin{center}{Figure~\ref{fig: asset data} about here} \end{center}

Since it is unlikely that consumption is observed perfectly, we modify equation (\ref{campbell-consumption}) to include an iid noise term $\eta_t \sim N(0, \sigma^2_\eta)$. The restrictions placed on the joint distribution of consumption and asset price growth by the structural model allow us to identify this noise term separately from the consumption innovation $\nu_t$. We use a fairly tight prior distribution to constrain the variance of $\eta_t$,  because the structural model would otherwise struggle to improve on the assumption that consumption and price-dividend ratio growth are both iid. This exercise is intended to illustrate the ADPF estimation method, rather than provide a precise explanation of intertemporal saving decisions, and while the addition of external habits greatly improves the consumption-based asset pricing model, its chief virtue is still its simplicity, rather than its flexibility.

Table~\ref{campbell-priors-table} gives the prior distributions of the other parameters. The priors are chosen to ensure that they imply reasonable properties for the risk-free interest rate and consumption. Specifically, we select priors on $g$ and $\sigma_\nu^2$ to ensure that the underlying consumption growth series is close in mean and variance to the observed one. More importantly, \cite{campbell_by_1999} show that the implied level of the risk-free rate is
\begin{equation}
r^f = -\log \overline{\beta} + \gamma g - \left(\frac{\gamma}{\overline{S}}\right)^2 \frac{\sigma^2}{2}\ .
\label{campbell-rf-eqn}
\end{equation}
Instead of placing a prior on $\overline{\beta}$ directly, we choose a prior for $r^f$ that ensures it is a low positive number. Finally, we choose relatively loose priors for $\gamma$ and $\phi$.
Table~\ref{campbell-priors-table} summarises the prior distributions of the parameters,  which we assume are mutually independent.

\begin{center}{Table~\ref{campbell-priors-table}  about here.} \end{center}

To estimate the model, we took 25,000 adaptive Metropolis Hastings draws, with the proposal covariance matrix initialised to a small diagonal matrix, and adaptation beginning after 100 draws. Unlike the previous example, the posterior mode located using the Kalman filter and a first-order approximation to the model was not particularly close to the posterior mode using a higher-order approximation. For that reason, we initialised the chain at values close to the calibrations described in \cite{campbell_by_1999}. In implementing the ADPF, we used the same approach as in previous examples. In all cases, we discarded the first 1000 MCMC draws.

\subsubsection{Results and Analysis}

 Table~\ref{campbell-param-results-table} reports the posterior means and standard deviations (in brackets) of the parameters. Table~\ref{campbell-inefficiencies-table} reports the acceptance rates an and the inefficiency estimates of the parameters.
The table shows that the Metropolis-Hastings acceptance rates generally increase and the inefficiencies generally decrease
with a higher number of particles for both the standard particle filter and the ADPF. Broadly speaking, the performance of the ADPF with a given number of particles appears to be comparable to that of the standard particle filter with 15 or 20 times more particles.

Table~\ref{campbell-ct-table} presents the computing times, calculated in the same manner as described above. Here, a clearer difference emerges between the standard particle filter and the ADPF. The estimated computing times of the ADPF are roughly twice as fast as the standard particle filter's. In calculating these times, we penalised the use of the analytical derivative of equation (\ref{cambell-third-order-eqn}) equally as heavily as using the law of motion itself. If these analytical derivatives are discounted---as may happen in applications where the derivatives are considerably simpler than the transition equations---the performance of the ADPF is around 5 times better than the standard particle filter's.

\begin{center}{Tables~\ref{campbell-param-results-table}, \ref{campbell-inefficiencies-table}  and \ref{campbell-ct-table}
about here.}
\end{center}

Notably, the posterior mean of  $r^f$ is unchanged from its prior mean. However, despite the apparently weak identification of this parameter, the nonlinear estimation  reveals some amount of information about it. Figure~\ref{fig: campbell plots} is a scatterplot showing the values of $\overline{\beta}$ implied by the draws of $r^f$ against the draws of $\phi$. (The graph shows 1500 randomly-selected draws from the MCMC chain for the ADPF with 50 particles.) Perhaps surprisingly, a more persistent habit stock (higher $\phi$) is associated with more value placed on the future (higher $\overline{\beta}$). This is in fact consistent with the relationship implied by equation (\ref{campbell-euler-eqn}). To see this, substitute (\ref{campbell-habit}) into (\ref{campbell-euler-eqn}) and ignore shocks after period $t$,
\[ \frac{P_t}{D_t} \propto \overline{\beta} \exp\left[ \gamma(1-\phi)\widetilde{s}_t  \right]\left(1 + \frac{P_{t+1}}{D_{t+1}}\right),  \]
showing that a rise in $\overline{\beta}$ is, other things equal, associated with a fall in $(1-\phi)$.

\begin{center}{Figure~\ref{fig: campbell plots} about here}\end{center}

\section{Conclusion}
\label{section: conclusion}

The filter discussed in this paper offers an attractive alternative to the standard particle filter for estimating nonlinear structural models. Broadly speaking, in comparison to the standard filter, the ADPF requires much lower computing times for a given level of estimation accuracy.

\begin{appendices}
\section{Proofs} \label{Proof: Thm 1}
The proof can be derived from \cite{delmoral:2004} (Section 7.4.2, Proposition 7.4.1). 
However, we believe that it is easier to
follow the proof of Theorem 1 in \cite{pitt_properties_2012} and we do so here. Let $\swarm_t = \{(x_t^k,\pi_t^k), k=1, \dots, N\}$ be the swarm of particles at time $t$.

\begin{proof}[Proof of Theorem  \ref{Thm: unbiased}]
\begin{align}
\Exp \left ( \phat_N(y_t|y_{1:t-1}) |\swarm_{t-1} \right ) & = \sum_{k=1}^N p(y_t|x_{t-1}^k) \pi_{t-1}^k \label{eq: one step}\\
\Exp\left ( \phat_N ( y_{t-h:t}|y_{1:t-h-1}) |\swarm_{t-h-1}\right ) & = \sum_{k=1}^N p(y_{t-h:t}|x_{t-h-1}^k) \pi^k_{t-h-1} \label{eq: multistep}\\
\Exp \left ( \phat_N ( y_{1:t}) |\swarm_0 \right ) & = \sum_{k=1}^N p(y_{1:t}|x_{0}^k) \pi^k_{0} \label{eqn: at zero}\\
\Exp \left ( \phat_N ( y_{1:t} ) \right ) & = p(y_{1:t})  \label{eqn: likelihood}
\end{align}
Equations \eqref{eq: one step} and \eqref{eq: multistep} are obtained as in Lemmas 6 and 7 respectively of  \cite{pitt_properties_2012}.
Equation \eqref{eqn: at zero} is obtained by taking $h = t-1$ in  \eqref{eq: multistep} and  \eqref{eqn: likelihood} is obtained
from  \eqref{eqn: at zero} because $\swarm_0 = \{(x_0^k,\pi_0^k), k=1, \dots, N\}$ with $x_0^k \sim p(x_0)$ and $\pi_0^k = 1/N$.
\end{proof}

\section{Optimisation Method}\label{app: optimization}

While the optimization of $\ell (u^k_t | y_t, x_{t-1}^k)$ can be performed in many ways, we find that the Levenberg-Marquardt algorithm \citep{marquardt_algorithm_1963}, as implemented by \cite{more_user_1980}, delivers satisfactory results. This algorithm is applied sequentially to each particle at each time step. By default, it uses numerical differentiation to estimate $\widehat{g} = \frac{\partial \ell}{\partial u_i}$ and $\widehat{A} = \frac{\partial^2 \ell}{\partial u_i \partial u_j}$. In the case of the asset pricing model,  analytical derivatives are easy to calculate and
are used instead. At each step of the iteration, the proposed value of $u$ is
 \[ u^{(i)} = u^{(i-1)} - \left( \widehat{A} + \nu I \right)^{-1}\widehat{g}. \]
The value of $\nu$ is initialised at 10, then increased by a factor of 10 if the proposed value is rejected, and decreased by a factor of 10 if the proposed value is accepted. The algorithm is deemed to have converged if $\vectornorm{\widehat{g}} < 10^{-3}$, where
$\vectornorm{\widehat{g}}$ is the Euclidean norm, or if the sum of squared residuals (scaled by their standard deviations) is less than $10^{-5}$. The algorithm is terminated after a maximum of 10 iterations.

In each case, each component of $u^k_t$ is initialised with a draw from $N(0,2)$ (recall that each $u$ variate is, by construction, standard normal). We made this choice to ensure that the algorithm explores the tails of the distribution with reasonable probability.

\section{Asset Pricing Approximation}\label{A: asset pricing approximation}

Taking a third-order Taylor approximation of equation (\ref{campbell-euler-eqn}), then evaluating the coefficients with $\widetilde{s}_t = \widetilde{s}_{t+1} = b_t = \nu_{t+1} = 0$, gives the following values for the $\Gamma_i$ coefficients in equation (\ref{cambell-third-order-eqn}), where $G = \exp (g)$, where $g$ is the growth rate from equation~\eqref{campbell-consumption}:

\begin{align*} \Gamma_0 &= \frac{{\overline{\beta}} G}{G^\gamma - {\overline{\beta}} G} , \\
 \Gamma_1 &= \frac{(1-\phi)G^{\gamma+1}{\overline{\beta}} \gamma}{(G^\gamma - {\overline{\beta}} G)(G^\gamma - \phi {\overline{\beta}} G)},   \\
\Gamma_2 &= \frac{1}{2}\Gamma_0 \Gamma_1 \frac{(1-\phi)(G^\gamma - {\overline{\beta}} G)\gamma(G^\gamma + \phi {\overline{\beta}} G)}{{\overline{\beta}} G(G^\gamma - \phi {\overline{\beta}} G)} ,  \\
\Gamma_3 &= \frac{1}{6}\frac{(1-\phi)^3 G^{\gamma+1} {\overline{\beta}} \gamma^3(G^{2\gamma} + 2 {\overline{\beta}} \phi G^{\gamma+1} + 2 {\overline{\beta}} \phi^2 G^{\gamma+1} + {\overline{\beta}}^2 \phi^3 G^2)}{(G^\gamma - {\overline{\beta}} G)(G^\gamma - \phi {\overline{\beta}} G)(G^\gamma - \phi^2 {\overline{\beta}} G)(G^\gamma - \phi^3 {\overline{\beta}} G)} .
\end{align*}

\end{appendices}


\bibliographystyle{asa}
\bibliography{adpf1}

\begin{thebibliography}{23}
\newcommand{\enquote}[1]{``#1''}
\expandafter\ifx\csname natexlab\endcsname\relax\def\natexlab#1{#1}\fi

\bibitem[{Amisano and Tristani(2010)}]{amisano_euro_2010}
Amisano, G. and Tristani, O. (2010), \enquote{Euro area inflation persistence
  in an estimated nonlinear {DSGE} model,} \textit{Journal of Economic Dynamics
  and Control}, 34, 1837--1858.

\bibitem[{Andreasen(2011)}]{andreasen_non-linear_2011}
Andreasen, M.~M. (2011), \enquote{Non-linear {DSGE} models and the optimized
  central difference particle filter,} \textit{Journal of Economic Dynamics and
  Control}, 35, 1671--1695.

\bibitem[{Andrieu et~al.(2010)Andrieu, Doucet, and
  Holenstein}]{andrieu_particle_2010}
Andrieu, C., Doucet, A., and Holenstein, R. (2010), \enquote{{Particle Markov
  chain Monte Carlo methods},} \textit{Journal of the Royal Statistical
  Society: Series B}, 72, 269--342.

\bibitem[{Aruoba et~al.(2011)Aruoba, Bocola, and Schorfheide}]{aruoba_new_2011}
Aruoba, S.~B., Bocola, L., and Schorfheide, F. (2011), \enquote{A new class of
  nonlinear time series models for the evaluation of {DSGE} models,} Mimeo,
  University of Maryland.

\bibitem[{Bouakez et~al.(2005)Bouakez, Cardia, and
  {Ruge-Murcia}}]{bouakez_habit_2005}
Bouakez, H., Cardia, E., and {Ruge-Murcia}, F.~J. (2005), \enquote{Habit
  formation and the persistence of monetary shocks,} \textit{Journal of
  Monetary Economics}, 52, 1073--1088.

\bibitem[{Campbell and Cochrane(1999)}]{campbell_by_1999}
Campbell, J.~Y. and Cochrane, J.~H. (1999), \enquote{By Force of Habit: A
  {Consumption-Based} Explanation of Aggregate Stock Market Behavior,}
  \textit{Journal of Political Economy}, 107, 205--251.

\bibitem[{Del~Moral(2004)}]{delmoral:2004}
Del~Moral, P. (2004), \textit{Feynman-Kac Formulae: Genealogical and
  Interacting Particle Systems with Applications}, New York: Springer.

\bibitem[{{Fern\'{a}ndez-Villaverde} and
  {Rubio-Ram\'{i}rez}(2008)}]{fernandez-villaverde_how_2007}
{Fern\'{a}ndez-Villaverde}, J. and {Rubio-Ram\'{i}rez}, J.~F. (2008),
  \enquote{How Structural Are Structural Parameters?} in \textit{{NBER}
  Macroeconomics Annual 2007}, eds. Acemoglu, D.~Rogoff, K. and M., W.,
  University of Chicago Press, vol.~22 of \textit{National Bureau of Economic
  Research Working Paper Series}.

\bibitem[{Gomme(2011)}]{gomme_second_order_2011}
Gomme, P.and~Klein, P. (2011), \enquote{Second-order approximation of dynamic
  models without the use of tensors,} \textit{Journal of Economic Dynamics and
  Control}, 35, 604--615.

\bibitem[{Gordon et~al.(1993)Gordon, Salmond, and Smith}]{gordon_novel_1993}
Gordon, N., Salmond, D., and Smith, A. (1993), \enquote{Novel approach to
  {nonlinear/non-Gaussian} Bayesian state estimation,} \textit{Radar and Signal
  Processing, {IEE} Proceedings F}, 140, 107 --113.

\bibitem[{Heunis(2011)}]{heunis_innovations_2011}
Heunis, A.~J. (2011), \enquote{The Innovations Problem,} in \textit{Oxford
  Handbook of Nonlinear Filtering}, eds. Crisan, D. and Rozovskii, B., New
  York: Oxford University Press, pp. 425---49.

\bibitem[{Hyde and Sherif(2005)}]{hyde_dont_2005}
Hyde, S. and Sherif, M. (2005), \enquote{Don{\textquoteright}t break the habit:
  structural stability tests of consumption asset pricing models in the {UK},}
  \textit{Applied Economics Letters}, 12, 289--296.

\bibitem[{Klein(2000)}]{klein_using_2000}
Klein, P. (2000), \enquote{Using the generalized Schur form to solve a
  multivariate linear rational expectations model,} \textit{Journal of Economic
  Dynamics and Control}, 24, 1405--1423.

\bibitem[{Marquardt(1963)}]{marquardt_algorithm_1963}
Marquardt, D.~W. (1963), \enquote{An Algorithm for {Least-Squares} Estimation
  of Nonlinear Parameters,} \textit{{SIAM} Journal on Applied Mathematics},
  431--441.

\bibitem[{Mor\'{e} et~al.(1980)Mor\'{e}, Garbow, and
  Hillstrom}]{more_user_1980}
Mor\'{e}, J.~J., Garbow, B.~S., and Hillstrom, K.~E. (1980), \enquote{User
  Guide for {MINPACK-1},} Tech. Rep. {ANL-80-74}, Argonne National Laboratory,
  Argonne, {IL}.

\bibitem[{Murray et~al.(2012)Murray, Jones, and
  Parslow}]{murray_collapsed_2012}
Murray, L.~M., Jones, E.~M., and Parslow, J. (2012), \enquote{On Collapsed
  State-space Models and the Particle Marginal {Metropolis-Hastings} Sampler,}
  Tech. rep.

\bibitem[{Pitt and Shephard(1999)}]{pitt_filtering_1999}
Pitt, M.~K. and Shephard, N. (1999), \enquote{Filtering via Simulation:
  Auxiliary Particle Filters,} \textit{Journal of the American Statistical
  Association}, 94, 590--599.

\bibitem[{Pitt et~al.(2012)Pitt, Silva, Giordani, and
  Kohn}]{pitt_properties_2012}
Pitt, M.~K., Silva, R., Giordani, P., and Kohn, R. (2012), \enquote{{On some
  properties of Markov chain Monte Carlo simulation methods based on the
  particle filter},} \textit{Journal of Econometrics}, in press.

\bibitem[{Roberts and Rosenthal(2001)}]{roberts_optimal_2001}
Roberts, G. and Rosenthal, J.~S. (2001), \enquote{Optimal scaling for various
  {Metropolis-Hastings} algorithms,} \textit{Statistical Science}, 16,
  351--367.

\bibitem[{{Schmitt-Groh\'{e}}(2004)}]{schmitt-grohe_solving_2004}
{Schmitt-Groh\'{e}}, S.and~Uribe, M. (2004), \enquote{Solving dynamic general
  equilibrium models using a second-order approximation to the policy
  function,} \textit{Journal of Economic Dynamics and Control}, 28, 755--775.

\bibitem[{Shiller(2006)}]{shiller_irrational_2006}
Shiller, R.~J. (2006), \textit{Irrational Exuberance}, New York:
  {Currency/Doubleday}.

\bibitem[{van Der~Merwe et~al.(2001)van Der~Merwe, Doucet, de~Freitas, and
  Wan}]{van_der_merwe_unscented_2001}
van Der~Merwe, R., Doucet, A., de~Freitas, N., and Wan, E.~A. (2001),
  \enquote{The unscented particle filter,} in \textit{Advances in Neural
  Information Processing Systems 13}, eds. Leen, T.~K., Dietterich, T.~G., and
  Tresp, V., Cambridge, {MA}: {MIT} Press, pp. 584---90.

\bibitem[{Wan and van Der~Merwe(2000)}]{wan_unscented_2000}
Wan, E.~A. and van Der~Merwe, R. (2000), \enquote{The unscented Kalman filter
  for nonlinear estimation,} in \textit{Adaptive Systems for Signal Processing,
  Communications, and Control Symposium 2000. {AS-SPCC.} The {IEEE} 2000},
  {IEEE}, pp. 153--158.

\end{thebibliography}

\clearpage
\begin{figure}[!h]
\centering
\mbox{
\subfigure{\includegraphics[width=6cm]{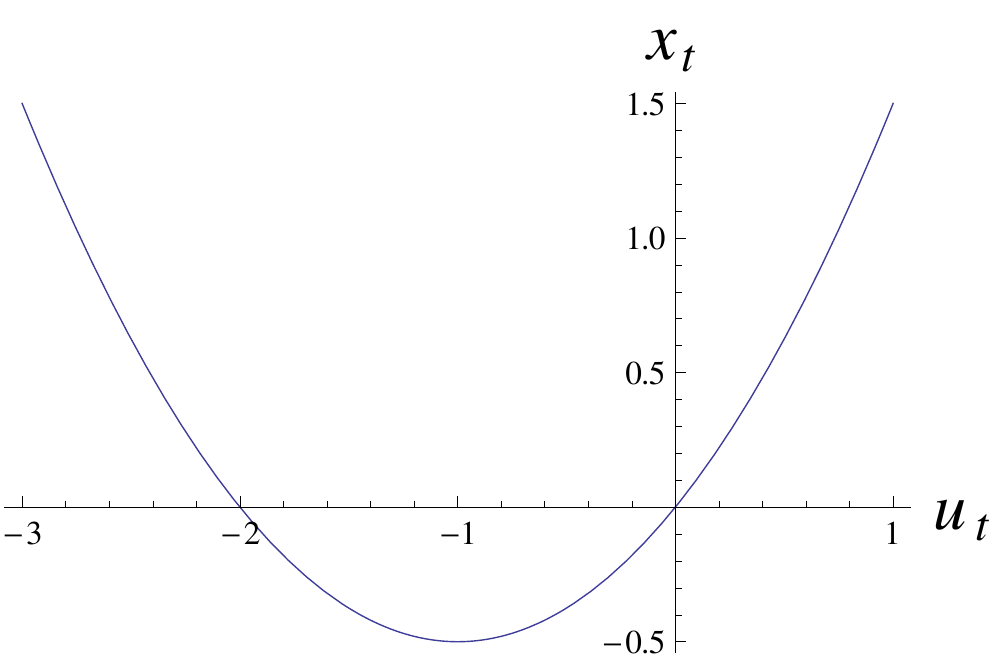}}\quad
\subfigure{\includegraphics[width=6cm]{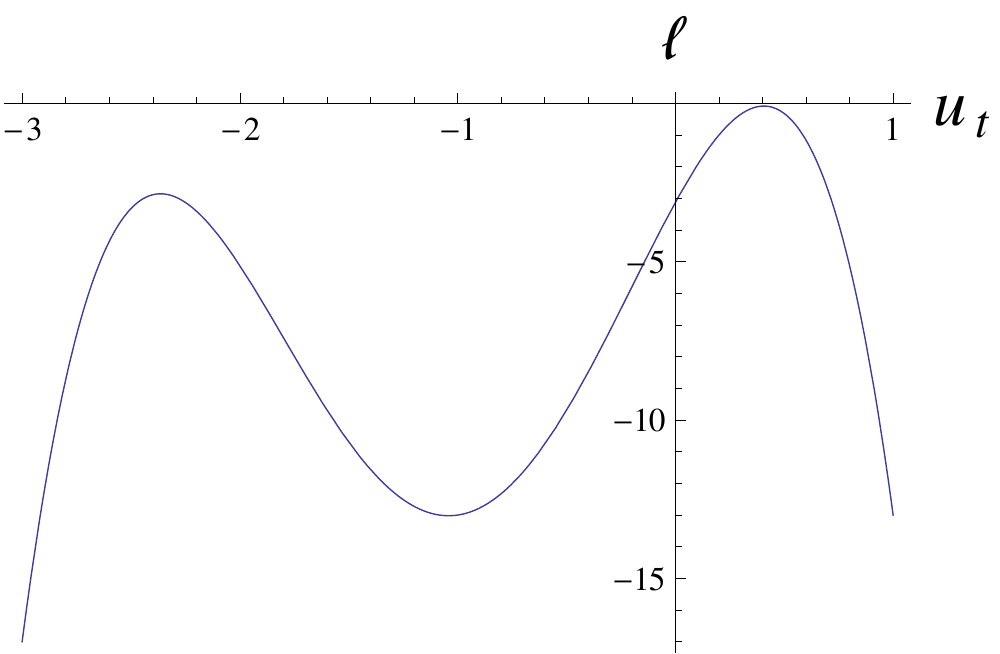} }}
\caption{New state $x_t$ as a function of shock $u_t$ (left) and log-posterior of $u_t$ (right) for a realisation of the quadratic AR(1) model}
\label{fig: bimodal densities}
\end{figure}

\begin{figure}[!h]
\centering
\mbox{
\subfigure{\includegraphics[width=3in]{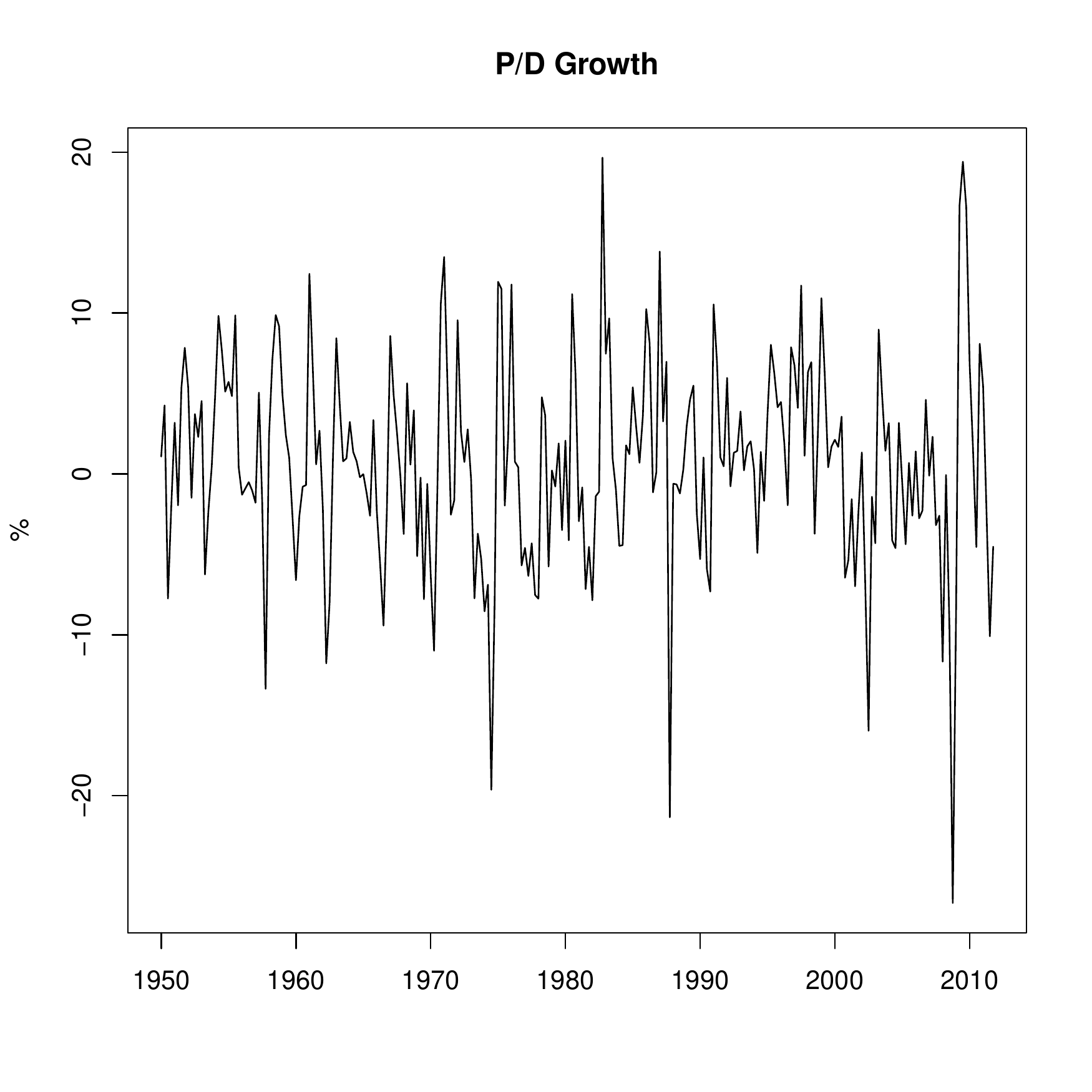}}\quad
\subfigure{\includegraphics[width=3in]{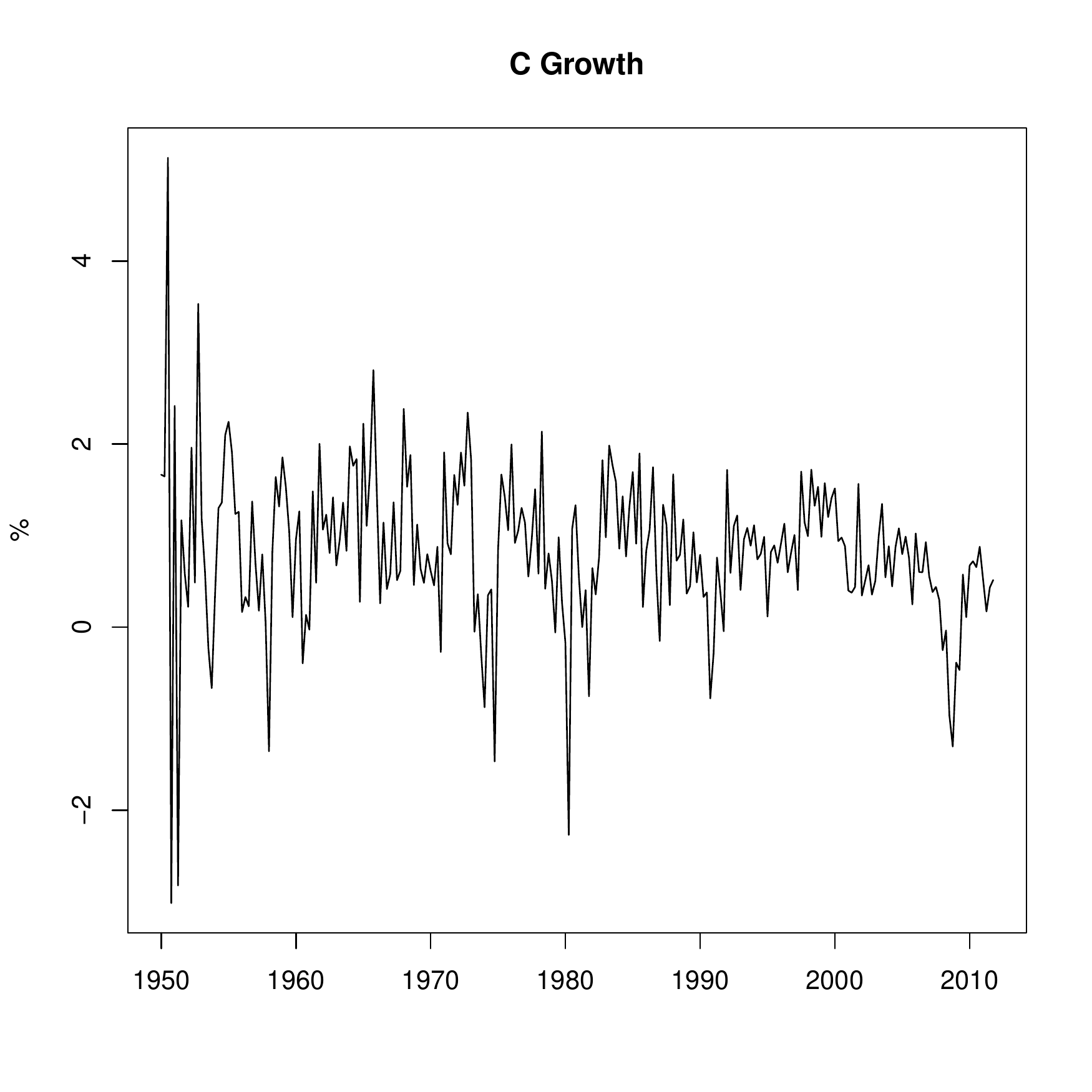} }
}
\caption{Time series plots of $P/D$ growth and $C$ growth used for the asset pricing model}
\label{fig: asset data}
\end{figure}

\begin{figure}[h!]
\centering
\includegraphics[width=4in]{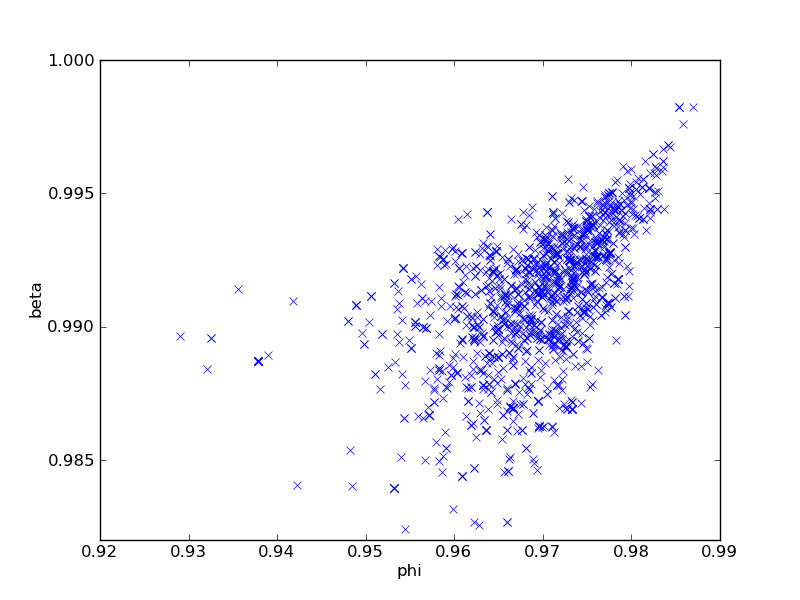}
\caption{Bivariate scatter-plot of the MCMC draws of $\beta$ and $\phi$ for the asset pricing model}
\label{fig: campbell plots}
\end{figure}

\begin{table}[!h]]
\centering
\begin{tabular}{c|ccc}
\hline
\rule{0pt}{15pt}$N$ particles & Median loglikelihood & IQR & Median std. dev. \\
\hline
\multicolumn{4}{c}{\rule{0pt}{11pt}Standard Particle Filter} \\
\hline
\rule{0pt}{11pt}100 & -3241.7 & 2876.8 & 4028.3  \\
500 & -285.8 & 588.7 & 847.8 \\
1000 & -118.8 & 500.9 & 334.4 \\
2000 & -80.0 & 206.9 & 104.1 \\
5000 & -69.6 & 35.1 & 21.7 \\
7500 & -68.3 & 37.7 & 9.0 \\
15000 & -67.7 & 6.4 & 2.3 \\
\hline
\multicolumn{4}{c}{\rule{0pt}{11pt}CUPF1} \\
\hline
\rule{0pt}{11pt}
50 & -7587.2 & 3203.8 & 6278.5 \\
100 & -3255.3 & 1881.4 & 4197.5 \\
150 & -1924.8 & 1278.6 & 3048.9 \\
\hline
\multicolumn{4}{c}{\rule{0pt}{11pt}Auxiliary Disturbance Particle Filter} \\
\hline
\rule{0pt}{11pt}
50 & -67.2 & 0.15 & 0.51  \\
\hline
\end{tabular}
\caption{Quadratic AR(1) model. Low nonlinearity, high \SNR}
\label{low-non-high-snr-table}
\end{table}

\begin{table}[!h]
\begin{tabular}{c|ccc}
\hline
\rule{0pt}{15pt}$N$ particles & Median loglikelihood & IQR & Median std. dev. \\
\hline
\multicolumn{4}{c}{\rule{0pt}{11pt}Standard Particle Filter} \\
\hline
\rule{0pt}{11pt}100 & -2062.7 & 3722.3 & 3251.4  \\
500  & -122.5 & 545.8 & 240.4  \\
1000  & -56.0 & 98.5 & 53.8  \\
2000 & -40.3 & 49.9 & 11.7  \\
5000 & -37.3 & 4.8 & 2.5  \\
7500 & -36.8 & 1.84 & 1.55  \\
15000 & -36.5 & 0.43 & 0.91  \\
\hline
\multicolumn{4}{c}{\rule{0pt}{11pt}CUPF1} \\
\hline
\rule{0pt}{11pt}
50 & -7915.0 & 12404 & 9644.9  \\
100 & -2011.1 & 6099.1 & 3618.4  \\
150 & -960.3 & 3168.2 & 2149.3  \\
\hline
\multicolumn{4}{c}{\rule{0pt}{11pt}Auxiliary Disturbance Particle Filter} \\
\hline
\rule{0pt}{11pt}
50 & -38.2 & 0.3 & 1.2  \\
\hline
\end{tabular}
\caption{Quadratic AR(1) model. High nonlinearity, high \SNR.}
\label{high-non-high-snr-table}
\end{table}

\begin{table}[ph]
\centering
\begin{threeparttable}
\begin{tabular}{c|cc}
\hline
\rule{0pt}{15pt}$N$ particles & Loglikelihood variance & Loglikelihood bias  \\
\hline
\multicolumn{3}{c}{\rule{0pt}{11pt}Standard Particle Filter} \\
\hline
\rule{0pt}{11pt}100 & 0.3824 {\footnotesize(.3816,.3837)} & -0.190 \\
500 & 0.0809 {\footnotesize(0.0807,0.0812)}  & -0.028 \\
1000 & 0.0403 {\footnotesize(0.0402,0.0404)} & -0.001 \\
2000 & 0.0211 {\footnotesize(0.0211,0.0212)} & -0.006 \\
5000 & 0.0081 {\footnotesize(0.00808,0.00813)} & 0.010 \\
7500 & 0.0052 {\footnotesize(0.00519,0.00522)} & 0.007 \\
15000 & 0.0027 {\footnotesize(0.00269,0.00271)} & 0.010 \\
\hline
\multicolumn{3}{c}{\rule{0pt}{11pt}CUPF1} \\
\hline
\rule{0pt}{11pt}
50 & 0.8731 {\footnotesize(0.8712,0.8761)} & -0.416 \\
100 & 0.3639 {\footnotesize(0.3631,0.3652)} & -0.196 \\
150 & 0.2869 {\footnotesize(0.2863,0.2879)} & -0.121 \\
\hline
\multicolumn{3}{c}{\rule{0pt}{11pt}Auxiliary Disturbance Particle Filter} \\
\hline
\rule{0pt}{11pt}
50 & 0.1076 {\footnotesize(0.1074,0.1080)}  & -0.117  \\
\hline
\end{tabular}
\caption{Quadratic AR(1) model.  Low nonlinearity, low \SNR. The
95\% confidence intervals for the estimated variances are  in brackets.}
\label{low-non-low-snr-bias-table}
\end{threeparttable}
\end{table}

\begin{table}[!h]
\centering
\begin{threeparttable}
\begin{tabular}{c|cc}
\hline
\rule{0pt}{15pt}$N$ particles & Loglikelihood variance & Loglikelihood bias  \\
\hline
\multicolumn{3}{c}{\rule{0pt}{11pt}Standard Particle Filter} \\
\hline
\rule{0pt}{11pt}100 & 14.99 {\footnotesize (14.96,15.04)} & -2.90 \\
500 & 0.8199 {\footnotesize (0.8181,0.8228)} & -0.30 \\
1000 & 0.3026 {\footnotesize (0.3020,0.3037)} & -0.13 \\
2000 & 0.1478 {\footnotesize(0.1475,0.1483)} & -0.05 \\
5000 & 0.0554 {\footnotesize(0.0553,0.0556)} & -0.001 \\
7500 & 0.0354 {\footnotesize(0.0353,0.0355)}  & -0.008 \\
15000 & 0.0168 {\footnotesize(0.0168,0.0169)} & 0.006 \\
\hline
\multicolumn{3}{c}{\rule{0pt}{11pt}CUPF1} \\
\hline
\rule{0pt}{11pt}
50 & 43.4 {\footnotesize(43.3,43.6)} & -6.50 \\
100 & 16.6 {\footnotesize(16.6,16.7)} & -3.06 \\
150 & 8.17 {\footnotesize(8.16,8.20)} & -1.66 \\
\hline
\multicolumn{3}{c}{\rule{0pt}{11pt}Auxiliary Disturbance Particle Filter} \\
\hline
\rule{0pt}{11pt}
50 & 0.623 {\footnotesize(0.622,0.626)} & -0.57 \\
\hline
\end{tabular}
\caption{Quadratic AR(1) model.  High nonlinearity, low \SNR.
95\% confidence intervals for the estimated variances are  in brackets.}
\label{high-non-low-snr-bias-table}
\end{threeparttable}
\end{table}

\begin{table}[!h]
\centering
\begin{threeparttable}
\begin{tabular}{c|cc}
\hline
\rule{0pt}{15pt}$N$ particles & Loglikelihood variance & Loglikelihood bias  \\
\hline
\multicolumn{3}{c}{\rule{0pt}{11pt}Standard Particle Filter} \\
\hline
\rule{0pt}{11pt}100 & 1.62$\times 10^7$ {\footnotesize(1.61$\times 10^7$,1.63$\times 10^7$)}& -4396.2 \\
500 & 7.19$\times 10^5$ {\footnotesize(7.17$\times 10^5$,7.21$\times 10^5$)}& -568.6 \\
1000 & 1.12$\times 10^5$ {\footnotesize(1.116$\times 10^5$,1.122$\times 10^5$)} & -198.6 \\
2000 & 1.083$\times 10^4$ {\footnotesize(1.081$\times 10^4$,1.087$\times 10^4$)} & -53.6 \\
5000 & 470.3 {\footnotesize(469.2,471.9)} & -9.1 \\
7500 & 81.8 {\footnotesize(81.7,82.1)} & -3.8 \\
15000 & 5.33 {\footnotesize(5.32,5.35)} & -1.0  \\
\hline
\multicolumn{3}{c}{\rule{0pt}{11pt}CUPF1} \\
\hline
\rule{0pt}{11pt}
50 & 3.94$\times 10^7$ {\footnotesize(3.93$\times 10^7$,3.96$\times 10^7$)} & -8611.5 \\
100 & 1.762$\times 10^7$  {\footnotesize(1.758$\times 10^7$,1.768$\times 10^7$)} & -4538.7 \\
150 & 9.30$\times 10^6$  {\footnotesize(9.28$\times 10^6$,9.33$\times 10^6$)} & -2972.2 \\
\hline
\multicolumn{3}{c}{\rule{0pt}{11pt}Auxiliary Disturbance Particle Filter} \\
\hline
\rule{0pt}{11pt}
50 & 0.2607 {\footnotesize(0.2601,0.2616)} & -0.05 \\
\hline
\end{tabular}
\caption{Quadratic AR(1) model. Low nonlinearity, high \SNR.
95\% confidence intervals for the estimated variances are in brackets.}
\label{low-non-high-snr-bias-table}
\end{threeparttable}
\end{table}

\begin{table}[!h]
\centering
\begin{threeparttable}
\begin{tabular}{c|cc}
\hline
\rule{0pt}{15pt}$N$ particles & Loglikelihood variance & Loglikelihood bias \\
\hline
\multicolumn{3}{c}{\rule{0pt}{11pt}Standard Particle Filter} \\
\hline
\rule{0pt}{11pt}100 & 1.057$\times 10^7$ {\footnotesize(1.055$\times 10^7$,1.061$\times 10^7$)} & -3129.4 \\
500 & 5.781$\times 10^4$ {\footnotesize(5.768$\times 10^4$,5.801$\times 10^4$)} & -156.4  \\
1000 & 2897.9 {\footnotesize(2891.7,2908.0)} & -37.5  \\
2000 & 135.7 {\footnotesize(135.4,136.2)} & -7.8  \\
5000 & 6.09 {\footnotesize(6.07,6.11)} & -1.27  \\
7500 & 2.394 {\footnotesize(2.389,2.402)} & -0.57  \\
15000 & 0.821 {\footnotesize(0.820,0.824)} & -0.10  \\
\hline
\multicolumn{3}{c}{\rule{0pt}{11pt}CUPF1} \\
\hline
\rule{0pt}{11pt}50 & 9.30$\times 10^7$ {\footnotesize(9.28$\times 10^7$,9.33$\times 10^7$)} & -1.10$\times 10^4$ \\
100 & 1.309$\times 10^7$ {\footnotesize(1.306$\times 10^7$,1.314$\times 10^7$)} & -3194.6 \\
150 & 4.62$\times 10^6$ {\footnotesize(4.61$\times 10^6$,4.64$\times 10^6$)} & -1653.8 \\
\hline
\multicolumn{3}{c}{\rule{0pt}{11pt}Auxiliary Disturbance Particle Filter} \\
\hline
\rule{0pt}{11pt}
50 & 1.522 {\footnotesize(1.519,1.527)} & -1.90 \\
\hline
\end{tabular}
\caption{Quadratic AR(1) model. High nonlinearity, high \SNR.
95\% confidence intervals for the estimated variances are in brackets.}
\label{high-non-high-snr-bias-table}
\end{threeparttable}
\end{table}

\begin{table}[!h]
\centering
\begin{center}
\begin{tabular*}{0.7\textwidth}{l|lccc}
\hline
\rule{0pt}{15pt}Parameter & Distribution & Mean & Std. dev. \\
\hline
$\rho$ & Beta & 0.8 & 0.1 \\
$\alpha$  & Beta & 0.333 & 0.015 \\
$\delta$ (\%) & Gamma & 0.5 & 0.07 \\
$\sigma_\epsilon$ & Gamma & 0.01 & 0.01 \\
\hline
\end{tabular*}
\caption{Priors for the structural parameters for the  neoclassical growth model.}
\label{growth-priors-table}
\end{center}
\end{table}

\begin{table}[!h]
\centering
\begin{tabular}{c|ccccc}
\hline
\rule{0pt}{15pt}$N$ particles & Acceptance rate & \multicolumn{4}{c}{Inefficiencies} \\
\hline & & $\alpha$ & $\delta$ & $\rho$ & $\sigma_\epsilon$
\\
\hline
\multicolumn{6}{c}{\rule{0pt}{11pt}Standard Particle Filter} \\
\hline
\rule{0pt}{11pt}500 & 2.1 & 769.7 & 580.4 & 491.2 & 404.1 \\
1500 & 12.3 & 43.4 & 87.0 & 52.7 & 56.2 \\
5000 & 21.8 & 18.3 & 19.5 & 20.4 & 18.7 \\
10000 & 24.9 & 16.8 & 14.5 & 18.0 & 16.5 \\
\hline
\multicolumn{6}{c}{\rule{0pt}{11pt}Auxiliary Disturbance Particle Filter} \\
\hline
\rule{0pt}{11pt}30 & 15.2 & 41.1 & 27.5 & 31.2 & 31.8 \\
50 & 19.1 & 25.0 & 22.1 & 21.3 & 23.4 \\
75 & 21.0 & 23.3 & 20.3 & 20.5 & 20.4 \\
100 & 23.3 & 18.0 & 17.5 & 16.4 & 16.7 \\
150 & 24.6 & 22.5 & 16.1 & 17.5 & 18.8 \\
200 & 25.1 & 19.2 & 16.0 & 17.0 & 16.8 \\
300 & 25.9 & 15.7 & 13.5 & 14.5 & 16.9 \\
\hline
\end{tabular}
\caption{Metropolis-Hastings acceptance rates and inefficiencies of the parameter estimates for the growth model.}
\label{acc-ineff-results-table}
\end{table}

\begin{table}[!h]
\centering
\begin{threeparttable}
\begin{tabular}{c|ccccc}
\hline
\rule{0pt}{15pt}$N$ particles & Loglikelihood variance & \multicolumn{4}{c}{Computing time / $10\mathrm{e}^{5}$} \\
\hline & & $\alpha$ & $\delta$ & $\rho$ & $\sigma_\epsilon$
\\
\hline
\multicolumn{6}{c}{\rule{0pt}{11pt}Standard Particle Filter} \\
\hline
\rule{0pt}{11pt}500 & 24.03 & 192.4 & 145.1 & 122.8 & 101.0 \\
1500 & 2.03 & 32.5 & 65.3 & 39.5 & 42.2 \\
5000 & 0.68 & 45.8 & 48.7 & 51.1 & 46.7 \\
10000 & 0.28 & 84.0 & 72.3 & 90.2 & 82.3 \\
\hline
\multicolumn{6}{c}{\rule{0pt}{11pt}Auxiliary Disturbance Particle Filter} \\
\hline
\rule{0pt}{11pt}30 & 1.72 & 9.9 & 6.6 & 7.5 & 7.6 \\
50 & 1.06 & 10.0 & 8.8 & 8.5 & 9.4 \\
75 & 0.67 & 14.0 & 12.2 & 12.3 & 12.2 \\
100 & 0.39 & 14.4 & 14.0 & 13.1 & 13.4 \\
150 & 0.43 & 27.0 & 19.3 & 21.0 & 22.6 \\
200 & 0.25 & 30.8 & 25.6 & 27.1 & 26.9 \\
300 & 0.15 & 37.6 & 32.5 & 34.8 & 40.5 \\
\hline
\end{tabular}
\footnotesize
\caption{Loglikelihood variances and computing times for the neoclassical growth model.
The variances of the loglikelihoods are calculated at the true parameter values, $\alpha =  1/3$, $\rho = 0.8$, $\delta =  0.05$ and $\sigma_\epsilon^2 = 0.02^2$. The value of $k=16$ is used in calculating computing times.}
\end{threeparttable}
\label{ct-results-table}
\end{table}

\begin{table}[!h]
\centering
\begin{center}
\begin{threeparttable}
\begin{tabular*}{0.7\textwidth}{l|lccc}
\hline
\rule{0pt}{15pt}Parameter & Distribution & Mean & Std. dev. \\
\hline
$\gamma$ & Gamma & 2 & 0.5 \\
$g$ (\%) & Gamma & 1.9 & 0.15 \\
$r^f$ (\%) & Normal & 1.0 & 0.01 \\
$\phi$ & Beta & 0.8 & 0.1 \\
$\sigma_\nu$ (\%) & Gamma & 0.8 & 0.03 \\
$\sigma_\eta$ (\%) & Gamma & 0.1 & 0.03 \\
$\sigma_\epsilon$ (\%) & Gamma & 5 & 0.7 \\
\hline
\end{tabular*}
\caption{Priors for the structural parameters in the asset pricing model.
Values for $g$ and $r^f$ and their standard deviations are in annualised percentage terms.
The prior distribution of $r^f$ is truncated to have positive support.}
\label{campbell-priors-table}
\end{threeparttable}
\end{center}
\end{table}

 \begin{table}[!h]
 \centering
\begin{threeparttable}
\begin{tabular}{c|ccccccc}
\hline
\rule{0pt}{15pt}$N$ particles &  \multicolumn{7}{c}{Posterior Mean (Std. Dev. in brackets)} \\
\hline & $g$ (\%) & $\gamma$ & $\phi$ & $r^f$ (\%) & $\sigma_\epsilon$ (\%) & $\sigma_\eta$ (\%) & $\sigma_\nu$ (\%)
\\
\hline
\multicolumn{8}{c}{\rule{0pt}{11pt}Standard Particle Filter} \\
\hline
\rule{0pt}{11pt}200 & 2.9 (0.28)& 0.93 (0.356)& 0.97 (0.007)& 1.0 (0.01)& 7.4 (0.43)& 0.8 (0.07)& 0.8 (0.05)\\
500 & 2.8 (0.20)& 0.87 (0.334)& 0.97 (0.008)& 1.0 (0.01)& 7.5 (0.50)& 0.8 (0.06)& 0.8 (0.04)\\
1000 & 2.9 (0.22)& 0.94 (0.399)& 0.97 (0.009)& 1.0 (0.01)& 7.3 (0.48)& 0.8 (0.08)& 0.8 (0.03)\\
\hline
\multicolumn{8}{c}{\rule{0pt}{11pt}Auxiliary Disturbance Particle Filter} \\
\hline
\rule{0pt}{11pt}30 & 2.9 (0.19)& 0.86 (0.309)& 0.97 (0.007)& 1.0 (0.01)& 7.4 (0.46)& 0.8 (0.06)& 0.8 (0.02)\\
50 & 2.8 (0.21)& 0.85 (0.309)& 0.97 (0.008)& 1.0 (0.01)& 7.4 (0.43)& 0.8 (0.06)& 0.8 (0.03)\\
\hline
\end{tabular}
\caption{Parameter estimates for the asset pricing model with standard errors in brackets. Values for $g$ and $r^f$ and their standard deviations are in annualised percentage terms.}
\label{campbell-param-results-table}
\end{threeparttable}
\end{table}

\begin{table}[!h]
\centering
\begin{tabular}{c|cccccccc}
\hline
\rule{0pt}{15pt}$N$ particles &  Acc. rate & \multicolumn{7}{c}{Inefficiencies} \\
 & & $g$ & $\gamma$ & $\phi$ & $r^f$ & $\sigma_\epsilon$ & $\sigma_\eta$ & $\sigma_\nu$
\\
\hline
\multicolumn{9}{c}{\rule{0pt}{11pt}Standard Particle Filter} \\
\hline
\rule{0pt}{11pt}200 & 5.0 & 115.8 & 195.9 & 444.5 & 129.0 & 477.6 & 359.9 & 128.7 \\
500 & 11.4 & 59.7 & 63.1 & 58.7 & 63.1 & 48.9 & 52.9 & 34.3 \\
1000 & 13.9 & 39.6 & 42.7 & 48.3 & 42.6 & 57.8 & 30.5 & 55.5 \\
\hline
\multicolumn{9}{c}{\rule{0pt}{11pt}Auxiliary Disturbance Particle Filter} \\
\hline
\rule{0pt}{11pt}30 & 11.4 & 52.2 & 38.0 & 50.0 & 112.5 & 49.3 & 51.9 & 54.9 \\
50 & 14.1 & 36.5 & 85.0 & 98.7 & 58.2 & 48.4 & 35.7 & 47.9 \\
\hline
\end{tabular}
\caption{MCMC parameter inefficiencies for the asset pricing model.}
\label{campbell-inefficiencies-table}
\end{table}

\begin{table}[!h]
\centering
\begin{threeparttable}
\begin{tabular}{c|ccccccc}
\hline
\rule{0pt}{15pt}$N$ particles &  \multicolumn{6}{c}{Computing Times} \\
\hline & $g$ & $\gamma$ & $\phi$ & $r^f$ & $\sigma_\epsilon$ & $\sigma_\eta$ & $\sigma_\nu$
\\
\hline
\multicolumn{7}{c}{\rule{0pt}{11pt}Standard Particle Filter} \\
\hline
\rule{0pt}{11pt}200  & 57.4 & 97.2 & 220.5 & 64.0 & 236.9 & 178.5 & 63.8 \\
500  & 74.0 & 78.2 & 72.8 & 78.2 & 60.6 & 65.6 & 42.5 \\
1000  & 98.3 & 105.8 & 119.8 & 105.6 & 143.2 & 75.6 & 137.7 \\
\hline
\multicolumn{7}{c}{\rule{0pt}{11pt}Auxiliary Disturbance Particle Filter} \\
\hline
\rule{0pt}{11pt}30  & 27.6 & 20.1 & 26.4 & 59.4 & 26.0 & 27.4 & 29.0 \\
50  & 32.1 & 74.8 & 86.9 & 51.3 & 42.6 & 31.4 & 42.2 \\
\hline
\end{tabular}
\caption{MCMC computing times for the parameters of the asset pricing model using the factor  $k = 7.1$ for the ADPF. }
\label{campbell-ct-table}
\end{threeparttable}
\end{table}

\end{document}